\documentclass{amsart}
\usepackage{amsmath,amssymb,mathabx}
\usepackage[utf8]{inputenc}
\usepackage{stmaryrd}
\usepackage{tikz}
\usepackage{hyperref}

\newcommand{\rank}{\mathrm{rank}}
\newcommand{\mergh}{\mathop{\varodot}}
\newcommand{\mergv}{\mathop{\varobar}}

\newtheorem{theorem}{Theorem}
\numberwithin{theorem}{section}
\newtheorem{example}[theorem]{Example}

\title{Optimal top dag compression}

\author[M.~Lohrey]{Markus Lohrey}
\author[C.~P.~Reh]{Carl Philipp Reh}
\author[K.~Sieber]{Kurt Sieber}
\email{\{lohrey.reh,sieber\}@eti.uni-siegen.de}
\thanks{This work has been supported by the DFG research project
LO 748/10-1 (QUANT-KOMP)}

\begin{document}

\begin{abstract}
It is shown that for a given ordered node-labelled tree of size $n$ and with
$s$ many different node labels, one can  construct in linear time a top dag
of height $O(\log n)$ and size $O(n / \log_\sigma n) \cap O(d \cdot \log n)$,
where $\sigma  = \max\{ 2, s\}$ and $d$ is the size of the minimal dag.
The size bound $O(n / \log_\sigma n)$ is optimal and improves on previous bounds.
\end{abstract}

\maketitle

\section{Introduction}
Top dags were introduced by Bille et al.~\cite{BilleGLW15} as a formalism for the compression of unranked node-labelled
trees. Roughly speaking, the top dag for a tree $t$ is the dag-expression of an expression that evaluates to $t$, where
the expression builds $t$ from edges using two merge operations (horizontal and vertical merge).
In \cite{BilleGLW15}, a linear time algorithm is presented that constructs from a tree of size $n$ with node labels
from the set $\Sigma$ a top dag of size $O(n / \log^{0.19}_\sigma n)$, where $\sigma = \max\{ 2, |\Sigma|\}$ (note that this definition of $\sigma$ avoids the base 1 in the logarithm). Later, in  \cite{Hubschle-Schneider15} this
bound was improved to $O(n \log\log n / \log_\sigma n)$ (for the same algorithm as in \cite{BilleGLW15}).
It is open whether this bound can be improved to $O(n / \log_\sigma n)$ for the construction from~\cite{BilleGLW15}.
A simple counting argument shows that $O(n / \log_\sigma n)$ is the information-theoretic lower bound for the size
of the top dag.
We present a new linear-time top dag construction that achieves this bound.
In addition, our construction has two properties that are also true for the original
construction of Bille et al.~\cite{BilleGLW15}: (i) the size of the constructed top dag is bounded by
$O(|\mathrm{dag}(t)| \log n)$, where $\mathrm{dag}(t)$ is the minimal dag of $t$ and (ii)
the height of the top dag is bounded by $O(\log(n))$.
Concerning (i) it was shown in \cite{BilleFG17} that the $\log(n)$-factor is unavoidable.
The logarithmic bound on the height
of the top dag in (ii) is important in order to get the logarithmic time bounds
for the querying operations (e.g., computing the label, parent node, first child, right sibling, depth, height, nearest common ancestor, etc.~of nodes given by their preorder numbers)
in~\cite{BilleGLW15}.

Our construction is based on a modification of the algorithm
BU-Shrink (bottom-up shrink) from~\cite{GanardiHJLN17}, which constructs in linear time
a TSLP of size $O(n/\log_\sigma n)$ for a given binary tree. In fact, we construct the top dag
in two phases: in the first step we apply the modification of BU-Shrink, whereas in a second phase
we apply the top dag construction of Bille et al.

\section{Preliminaries}

Let $\Sigma$ be a finite alphabet.
By $\mathcal{T}$ we denote the set of ordered labelled trees with labels from $\Sigma$.
Here, ``ordered'' means that the children of every node are linearly ordered.
Also note that trees are unranked in the sense that the node label does not determine
the number of children of a node, which is also called
the degree of the node.
For a tree $t \in \mathcal{T}$ we denote the label of its root by $r_t$,
the set of its nodes by $V(t)$ and the set of its edges by $E(t)$. By $\lambda_t(v)$ we
denote the label of the node $v$ in $t$.
For $v \in V(t)$ we denote with $t(v)$ the subtree of $t$ that is rooted in $v$.

A {\em cluster of rank $0$} is just a tree $t$. A {\em cluster of rank $1$} consists of a tree $t$ together with
a distinguished leaf node $\ell_t$ that we call the {\em bottom boundary node} of $t$.
In both cases, the root $r_t$ is called the {\em top boundary node}.
Let $C_i$ be the set of all clusters of rank $i \in \{0,1\}$ and let
$C = C_0 \cup C_1$.  With $\rank(t)$ we denote the rank of the cluster $t$.
An atomic cluster consists of a single edge, i.e., it is a tree with two nodes.

We define two partial binary merge operations $\mergv,\mergh : C
\times C \to C$:
\begin{enumerate}
\item $s \mergv t$ (the vertical merge of $s$ and $t$)
is only defined if $s \in C_1$ and $\lambda_s(\ell_s) = \lambda_t(r_t)$.
We obtain $s \mergv t$ by taking the disjoint union of $s$ and $t$ and then merging $\ell_s \in V(s)$ with $r_t \in V(t)$ (note that
this is possible since the labels coincide). The rank of $s \mergv t$ is $\rank(t)$ and if $t \in C_1$, then
the bottom boundary node of $s \mergv t$ is $\ell_t$.
\item $s \mergh t$ (the horizontal merge of $s$ and $t$)
is only defined if $\rank(s)+\rank(t)\leq 1$ and  $\lambda_s(r_s) = \lambda_t(r_t)$.
We obtain $s \mergh t$ by taking the disjoint union of $s$ and $t$ and then merging $r_s \in V(s)$ with $r_t \in V(t)$ (note that
this is possible since the labels coincide).
The rank of $s \mergh t$ is $\rank(s)+\rank(t)$.
In case $s \in C_1$ (resp., $t \in C_1$), the bottom boundary node of $s \mergh t$
is $\ell_s$ (resp., $\ell_t$).
\end{enumerate}
The (minimal) directed acyclic graph (dag) for a tree $t$ is obtained by keeping
for every subtree $t'$ of $t$ only one occurrence of that subtree and replacing
every edge that goes to a node in which a copy of $t'$ is rooted by an edge to
the unique chosen occurrence of $t'$. We denote this dag as $\mathrm{dag}(t)$.
Note that the number of nodes in $\mathrm{dag}(t)$ is the number of different
subtrees that occur in $t$.

We can now define {\em top trees} and {\em top dags}. A top tree is a binary node-labelled ordered tree,
where every internal node is labelled with one of the two operations $\mergv, \mergh$ and
every leaf is labelled with an atomic cluster plus a bit for the rank of the cluster. The latter information
can be represented by a triple $(a,b,i)$ with $a,b \in \Sigma$ and $i \in \{0,1\}$.
A top tree $T$ can be evaluated to a tree $t \in \mathcal{T}$ by recursively applying the merge operations
at its inner nodes (this might lead to an undefined value since the merge operations are only partially defined).
We say that $T$ is a top tree for $t$ and $\mathrm{dag}(T)$ is a top dag for $t$.

Let $t \in \mathcal{T}$ be a tree.
A {\em subcluster} of $t$ of rank one is an induced subgraph of $t$ that is obtained as follows:
Take a node $v \in V(t)$ with the ordered sequence of children $u_1,\ldots, u_d$
and let $1 \leq i \leq j \leq d$.
Let $u \in V$ be a node that belongs to one of the subtrees $t(u_i), \ldots, t(u_j)$.
Then the tree is induced by the nodes in $\{u,v\} \cup (\bigcup_{s = i}^j t(u_s)  \; \setminus \; t(u))$.
The node $u$ (resp., $v$) is the top (resp., boundary) node of the cluster.
A subcluster of $t$ of rank zero is obtained in the same way, except that we take the tree
induced by the nodes in $\{u\} \cup \bigcup_{s = i}^j t(u_s)$.
Its top boundary node is $u$.
Note that every edge of $t$ is a subcluster of $t$.
We identify a subcluster of $t$ with the set of edges of $t$ belonging to the subcluster.
If $T$ is a top tree for $t$ then it follows easily by induction
that every subtree of $T$ evaluates to an isomorphic copy of a subcluster of $t$.

\section{Optimal worst-case compression}

We can now state and prove the main result of this paper:

\begin{theorem} \label{thm-n/logn}
Let $\sigma = \max\{|\Sigma|, 2\}$.
There is a linear time algorithm that computes from a given tree $t \in \mathcal{T}$
with $n \geq 1$ edges a top dag of height $O(\log n)$, whose size is bounded
by $O(n / \log_\sigma n)$ and $O(|\mathrm{dag}(t)| \cdot \log n)$.
\end{theorem}

\begin{proof}
We first prove the theorem without the bound $O(|\mathrm{dag}(t)| \cdot \log n)$
on the size of the constructed top dag. In a second step, we explain how to modify
the algorithm in order to get the $O(|\mathrm{dag}(t)| \cdot \log n)$ bound.

Take a tree $t \in \mathcal{T}$ with $n \geq 1$ edges and let $\sigma = \max\{ |\Sigma|, 2 \}$.
We build from $t$ a sequence of trees $t_0, t_1, \ldots t_m$, where every edge $(u,v) \in E(t_i)$
($u$ is the parent node of $v$)
is labelled with a subcluster $c^i_{u,v}$ of $t$.
If $v$ is a leaf of $t_i$, then $c^i_{u,v}$ is a subcluster of rank $0$ with top boundary node $u$, otherwise
$c^i_{u,v}$ is a subcluster of rank $1$ with top boundary node $u$ and bottom boundary node $v$.
The number of edges in the subcluster $c^i_{u,v}$ is also called the weight $\gamma^i_{u,v}$ of the edge $(u,v)$.

Our algorithm does not have to store the subclusters explicitly but only their weights. Moreover, the
algorithm builds for every edge $(u,v) \in E(t_i)$ a top tree $T^i_{u,v}$.
The invariant of the algorithm is that $T^i_{u,v}$ evaluates to
(an isomorphic copy of) the subcluster
$c^i_{u,v}$. The top trees $T^i_{u,v}$ are stored as pointer structures, but below we write them for better readability
as expressions using the operators $\mergv$ and $\mergh$.

The initial tree $t_0$ is the tree $t$, where $c^0_{u,v} = \{(u,v)\}$ for
every edge $(u,v) \in E(t_0)$. This is a subcluster of rank $0$
if $v$ is a leaf, and of rank $1$ otherwise. Hence, we set
$\gamma^0_{u,v} = 1$ and
$T^0_{u,v} = (\lambda_t(u), \lambda_t(v), i)$,
where $i$ is the rank of subcluster $c^0_{u,v}$.

Let us now fix a number $k \leq n$ that will be made precise later. Our algorithm proceeds as follows:
Let $t_i$ be the current tree. We proceed by a case distinction. Ties between the following three
cases can be broken in an arbitrary way. The updating for the subclusters $c^i_{u,v}$ is only shown
to give a better intuition for the algorithm; it is not part of the algorithm.

\medskip
\noindent
{\em Case 1.}
 There exist edges $(u,v), (v,w) \in E(t_i)$ of weight at most $k$ such that $w$ is the unique child
of $v$. We obtain $t_{i+1}$ from $t_i$ by
({\it i}\/) removing the node $v$,
and ({\it ii}) replacing  the edges $(u,v), (v,w)$ by the edge $(u,w)$.
Moreover, we set
\begin{eqnarray*}
c^{i+1}_{u,w} & := & c^i_{u,v} \cup c^i_{v,w}, \\
T^{i+1}_{u,w} & := & T^i_{u,v} \mergv T^i_{v,w}, \\
\gamma^{i+1}_{u,w} &:=& \gamma^i_{u,v} + \gamma^i_{v,w} .
\end{eqnarray*}
For all edges
 $(x,y) \in E(t_i) \setminus \{(u,v), (v,w) \}$ we set $c^{i+1}_{x,y} := c^i_{x,y}$,
 $T^{i+1}_{x,y} := T^i_{x,y}$ and
 $\gamma^{i+1}_{x,y} := \gamma^i_{x,y}$.

\medskip
\noindent
{\em Case 2.} 
There exist edges $(u,v), (u,w) \in E(t_i)$ of weight at most $k$
such that $v$ is a leaf and the left sibling of $w$.
Then $t_{i+1}$ is obtained from $t_i$ by removing the edge $(u,v)$.
Moreover, we set
\begin{eqnarray*}
c^{i+1}_{u,w} &:=& c^i_{u,v} \cup c^i_{u,w},\\
T^{i+1}_{u,w} &:=& T^i_{u,v} \mergh T^i_{u,w}, \\
\gamma^{i+1}_{u,w} &:=& \gamma^i_{u,v} + \gamma^i_{u,w}.
\end{eqnarray*}
For all edges
$(x,y) \in E(t_i) \setminus \{(u,v), (u,w) \}$ we set $c^{i+1}_{x,y} := c^i_{x,y}$,
$T^{i+1}_{x,y} := T^i_{x,y}$ and
$\gamma^{i+1}_{x,y} := \gamma^i_{x,y}$.

\medskip
\noindent
{\em Case 3.} 
There exist edges $(u,v), (u,w) \in E(t_i)$ of weight at most $k$
such that $v$ is a leaf and the right sibling of $w$.
Then $t_{i+1}$ is obtained from $t_i$ by removing the edge $(u,v)$.
Moreover, we set
\begin{eqnarray*}
c^{i+1}_{u,w} &:=& c^i_{u,w} \cup c^i_{u,v}, \\
T^{i+1}_{u,w} &:=& T^i_{u,w} \mergh T^i_{u,v}, \\
\gamma^{i+1}_{u,w} &:=& \gamma^i_{u,w} + \gamma^i_{u,v} .
\end{eqnarray*}
For all edges
$(x,y) \in E(t_i) \setminus \{(u,v), (u,w) \}$ we set $c^{i+1}_{x,y} := c^i_{x,y}$,
$T^{i+1}_{x,y} := T^i_{x,y}$ and
$\gamma^{i+1}_{x,y} := \gamma^i_{x,y}$.

\medskip
\noindent
If none of the above three cases holds, then the algorithm stops. Let $t' = t_m$ be the final tree that we computed.
Note that every edge $(u,v)$ of $t'$ has weight at most $2k$.
We now bound the number of edges of $t'$:

\medskip
\noindent
{\em Claim:} The number of edges of $t'$ is at most $\frac{8n}{k}$

\medskip
\noindent
Let $n'$ be the number of edges of $t'$. Thus, $t'$ has $n'+1$ many nodes.
If $n' \leq 1$ we are done, since $\frac{8n}{k} \geq 8$.
So, assume that $n' \geq 2$. Let $U$ be the set of all nodes of degree at most one except for the root node.
We must have $|U| \geq n'/2$.
For every node $u \in U$, let $p(u)$ be its parent node.
We now assign to certain edges of $t'$ (possibly several) markings by doing the following for every
$u \in U$: If the weight of the edge
$(p(u),u)$ is larger than $k$ then we assign to $(p(u),u)$ a marking.
Now assume that the weight of $(p(u),u)$ is at most $k$.  If $u$ has degree one
and $v$ is the unique child of $u$, then the weight of $(u,v)$ must be larger than $k$
(otherwise, we would merge the edges $(p(u), u)$ and $(u,v)$), and we assign a marking
to $(u,v)$.
Let us now assume that $u$ is a leaf.
Since $t'$ has at least two edges, one of the following
three edges must exist:
\begin{itemize}
\item $(p(u), v)$, where $v$ is the left sibling of $u$,
\item $(p(u), v)$, where $v$ is the right sibling of $u$,
\item $(v,p(u))$ where $p(u)$ has degree one.
\end{itemize}
Moreover, one of these edges must have weight more than $k$.
We choose such an edge and assign a marking to it. The following then holds:
\begin{itemize}
\item Markings are only assigned to edges of weight more than $k$.
\item Every edge of $t'$ can get at most 4 markings.
\item In total, $t'$ contains $|U| \geq n'/2$ many markings.
\end{itemize}
Since the sum of all edge weights  of $t'$ is $n$, we obtain
\[
n \geq \frac{|U| \cdot k}{4} \geq     \frac{n' \cdot k}{8}
\]
Thus, we have $n' \leq \frac{8n}{k}$.

We now build a top tree $T$ for $t$ as follows:
Construct a top tree $T'$  for $t'$ of height $O(\log n)$ using the algorithm from
\cite{BilleGLW15}. Consider a leaf $e$ of $T'$.
It corresponds to an edge $(u,v) \in E(t')$. In the process of folding the cluster
$c^m_{u,v}$ into the edge $(u,v)$ we have constructed the top tree $T_e := T_{u,v}^m$
that evaluates to the cluster  $c^m_{u,v}$. Therefore, we obtain a top tree $T$ for $t$ by
replacing every leaf $e$ of $T'$ by the top tree $T_e$.
To bound the minimal dag of $T$ we have to count the number of different subtrees of $T$.
This number can be upper bounded by the number of nodes in $T'$ (which is in $O(n/k)$) plus the number
of different top trees of size at most $2k$. The latter number can be bounded as follows:
A top tree for a tree from $\mathcal{T}$ is a binary tree with $2(|\Sigma|^2 + 1)$ many node labels
($2 |\Sigma|^2$ many different atomic clusters together with the bit for their rank and two labels for the two merge operations).
The number of binary trees with $m$ nodes is bounded by $4^m$. Hence, we can bound the number
of different top trees of size at most $2k$ by $2k \cdot r^k$ with $r = 64(|\Sigma|^2 +1)^2$.
Note that $\log(r) \in \Theta(\log \sigma)$, where $\sigma = \max \{ |\Sigma|,2\}$.
Take $k = \frac{1}{2} \log_r(n) \in \Theta(\log_\sigma n)$. We obtain the following upper bound on the number of non-isomorphic subtrees of $T$:
\[
O\left( \frac{n}{\log_\sigma n}\right) + \log_r(n) \cdot \sqrt{n} = O\left(\frac{n}{\log_\sigma n}\right)
\]
Moreover, the height of $T$ is in $O(\log n)$ since $T'$ and all $T_e$ have height $O(\log n)$.

It remains to argue that our algorithm can be implemented in linear time. The arguments are more or less
the same as for the analysis of BU-Shrink in~\cite{GanardiHJLN17}: The algorithm maintains for every node of $t_i$ its degree,
and for every edge $(u,v)$ its weight $\gamma^i_{u,v}$. Additionally, the algorithm maintains
a queue that contains pointers to all edges $(u,v)$ of $t_i$ having weight at most $k$ and such that $v$ has
degree one. Then every merging step can be done in constant time, and there are at most $n$ merging steps.
Finally, the minimal dag of $T$ can be computed in linear time by~\cite{DoSeTa80}.

We now explain the modification of the above algorithm such that the constructed top dag
has size $O(|\mathrm{dag}(t)| \cdot \log n)$. The main idea is that we perform the first phase
of the above algorithm (where the tree $t'$ is constructed) on $\mathrm{dag}(t)$ instead of $t$ itself.
Thus, the algorithm starts with the construction of $d := \mathrm{dag}(t)$ from $t$, which is possible in linear time~\cite{DoSeTa80}.
We now build from $d$ a sequence of dags $d_0, d_1, \ldots d_m$, where analogously
to the above construction every edge $(u,v)$
of $d_i$ is labelled with a weight $\gamma^i_{u,v}$ and a top tree $T^i_{u,v}$.
Since we are working on the dag, we cannot assign a unique subcluster $c^i_{u,v}$ of $t$ to the dag-edge
$(u,v)$. In fact, every edge of $d_i$ represents a set of isomorphic subclusters that are shared in the dag. The
top tree $T^i_{u,v}$ evaluates to an isomorphic copy of these subclusters.

The initial dag $d_0$ is the dag $d$ where every edge $(u,v) \in E(d_0)$
is labelled with the top tree $T^0_{u,v}$ that only consists of the leaf $(\lambda_d(u), \lambda_d(v), i)$
($\lambda_d$ assigns to every node of the dag its label from $\Sigma$)
where $i=0$ if $v$ is a leaf of the dag and $i=1$ otherwise.
We take the same threshold value $k \leq n$ as before.
Let $d_i$ be the current dag. Also the case distinction is the same as before:

\medskip
\noindent
{\em Case 1.}
 There exist edges $(u,v), (v,w) \in E(d_i)$ of weight at most $k$ such that $w$ is the unique child
of $v$. We obtain $d_{i+1}$ from $d_i$ by replacing  the edge $(u,v)$ by the edge $(u,w)$.
If the node $v$ has no incoming edge after this modification, we can remove
 $v$ and the edge $(v,w)$ (although this is not necessary for the further arguments).
The weights $\gamma^i_{x,y}$ and the top trees $T^i_{x,y}$ are updated in exactly the same
way as in the previous case~1.

\medskip
\noindent
{\em Case 2.}
There exist edges $(u,v), (u,w) \in E(d_i)$ of weight at most $k$
such that $v$ is a leaf of $d_i$ (i.e., has no outgoing edge) and the left sibling of $w$.
Then $d_{i+1}$ is obtained from $d_i$ by removing the edge $(u,v)$.
If $v$ has no more incoming edges after this modification, then we can also remove $v$.
The weights $\gamma^i_{x,y}$ and the top trees $T^i_{x,y}$ are updated in exactly the same
way as in the previous case~2.

\medskip
\noindent
{\em Case 3.}
There exist edges $(u,v), (u,w) \in E(d_i)$ of weight at most $k$
such that $v$ is a leaf and the right sibling of $w$.
Then $d_{i+1}$ is obtained from $d_i$ by removing the edge $(u,v)$.
If $v$ has no more incoming edges after this modification, then we can remove $v$.
The weights $\gamma^i_{x,y}$ and the top trees $T^i_{x,y}$ are updated in exactly the same
way as in the previous case~3.

\medskip
\noindent
If none of the above three cases holds, then the algorithm stops. Let $d' = d_m$ be the final dag that we computed.
We can unfold $d'$ to a tree $t'$. This tree $t'$ is one of the potential outcomes of the above tree version of the
algorithm. The rest of the construction is the same as before, i.e.,
we apply the algorithm of Bille et al.~in order to get a top tree $T'$ for $t'$, which is then combined with the top trees
$T^m_{u,v}$ in order to get a top dag $T$ for $t$. Let $D$ be the minimal dag of $T$.
The size bound $O(n / \log_\sigma n)$ and the height bound $O(\log n)$
for  $D$ follow from our previous arguments. It remains to show that the size of $D$ is
bounded by $O(|\mathrm{dag}(t)| \cdot \log n)$. Note that the size of the dag $d'  = d_m$ is bounded
by $|\mathrm{dag}(t)|$ (the number of nodes and edges does not increase when constructing $d_{i+1}$
from $d_i$). Moreover, every top tree $T^m_{u,v}$ has size $O(\log_\sigma n) = O(\log n / \log \sigma)$.
Therefore, the total size of all top trees $T^m_{u,v}$ is bounded by
$O(|\mathrm{dag}(t)| \cdot \log n / \log \sigma)$. Moreover, the construction of Bille et al.~ensures that
the  size of the top dag for $t'$ is bounded by $O(|d'| \cdot \log |t'|) \leq O(|\mathrm{dag}(t)| \cdot \log n)$
(since $d'$ is a dag for $t'$). This implies that the size of the final top dag $D$ is bounded by $O(|\mathrm{dag}(t)| \cdot \log n)$.
\end{proof}

\begin{example} \rm
  \tikzstyle{vertex}=[circle,draw]
  Let $\Sigma = \{a,b\}$. In the following example we show three nodes of a dag
$d_0$ (on the left) and a possible run of the merging algorithm up to $d_2$ (on the right):
  \\ 
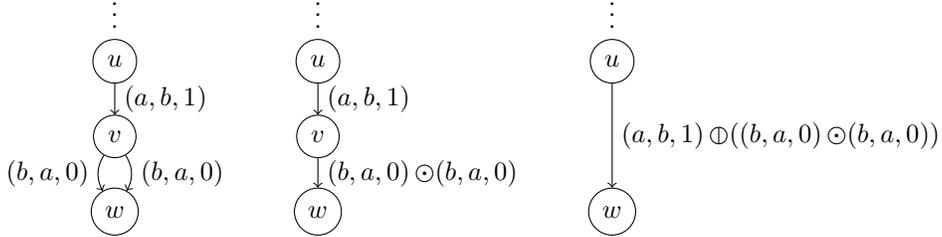
\noindent
\begin{center}
  \begin{tikzpicture}
    \node (A) at (0,3) {};
    \node[vertex] (B) at (0,2) {$u$};
    \node[vertex] (C) at (0,1) {$v$};
    \node[vertex] (D) at (0,0) {$w$};

    \path (A) -- (B) node[midway, sloped] {$\dots$};
    \draw [->](B) edge node[right] {$(a,b,1)$} (C);
    \draw [->](C) edge[bend left] node[right] {$(b,a,0)$} (D);
    \draw [->](C) edge[bend right] node[left] {$(b,a,0)$} (D);
  \end{tikzpicture}
  \qquad
  \begin{tikzpicture}
    \node (A) at (0,3) {};
    \node[vertex] (B) at (0,2) {$u$};
    \node[vertex] (C) at (0,1) {$v$};
    \node[vertex] (D) at (0,0) {$w$};

    \path (A) -- (B) node[midway, sloped] {$\dots$};
    \draw [->](B) edge node[right] {$(a,b,1)$} (C);
    \draw [->](C) edge node[right] {$(b,a,0) \mergh (b,a,0)$} (D);
  \end{tikzpicture}
  \qquad
  \begin{tikzpicture}
    \node (A) at (0,3) {};
    \node[vertex] (B) at (0,2) {$u$};
    \node[vertex] (D) at (0,0) {$w$};

    \path (A) -- (B) node[midway, sloped] {$\dots$};
    \draw [->](B) edge node[right] {$(a,b,1) \mergv ((b,a,0) \mergh (b,a,0))$} (D);
  \end{tikzpicture}
\end{center}
In the first step we merge the two atomic clusters $(b,a,0)$ using $\mergh$.
We can do this by removing the left edge $(v,w)$ or the right edge $(v,w)$,
since $w$ is a leaf. Next, we merge the cluster $(a,b,1)$ using $\mergv$.
This is done by removing the edge $(u,v)$ and replacing it with the edge $(u,w)$.
Since $v$ now has no incoming edges, $(v,w)$ is removed.
The weight of the edge $(u,w)$ is $3$ since its cluster is a tree with three
edges. This means that, in order to be able to perform these two merges, the
starting tree $t$ must have size at least
$2k(64(|\Sigma|^2+1)^2)^k$, which is $10240000$, since $k=2$ and $|\Sigma|=2$.
\end{example}

\end{document}